\documentclass[10pt,journal,compsoc]{IEEEtran}
\usepackage{graphicx}
\usepackage{subfigure}
\usepackage{amsmath}
\usepackage{booktabs}
\usepackage{algorithm}
\usepackage{color}
\usepackage{hyperref}
\usepackage{mathtools}
\usepackage{algpseudocode}
\usepackage{verbatim}
\usepackage{flushend}
\usepackage{amssymb}
\usepackage{amsthm}
\usepackage{multirow}
\usepackage{algorithmicx}
\usepackage{perpage} 
\MakePerPage{footnote} 
\newtheorem{theorem}{Theorem}
\DeclareMathOperator{\E}{\mathbb{E}}
\allowdisplaybreaks 
\algnewcommand\algorithmicinput{\textbf{Input:}}
\algnewcommand\INPUT{\item[\algorithmicinput]}
\algnewcommand\algorithmicoutput{\textbf{Output:}}
\algnewcommand\OUTPUT{\item[\algorithmicoutput]}
\algnewcommand\algorithmicbegin{\textbf{begin}}
\algnewcommand\BEGIN{\item[\algorithmicbegin]}
\algnewcommand\algorithmicendbegin{\textbf{end}}
\algnewcommand\ENDBEGIN{\item[\algorithmicendbegin]}

\algdef{SE}[DOWHILE]{Do}{doWhile}{\algorithmicdo}[1]{\algorithmicwhile\ #1}%

\usepackage{hyperref}

\ifCLASSOPTIONcompsoc
  \usepackage[nocompress]{cite}
\else
  \usepackage{cite}
\fi

\ifCLASSINFOpdf
\else
\fi
\newcommand{\liu}[1]{\textcolor{black}{#1}}
\newcommand{\liuS}[1]{\textcolor{black}{#1}}

\begin{document}
\newcommand{\TheName}{\texttt{\textbf{FedCube}}}

\title{Data Placement for Multi-Tenant Data Federation on the Cloud}


\author{Ji Liu\IEEEauthorrefmark{2}\IEEEauthorrefmark{9}\IEEEauthorrefmark{1},
        Lei Mo\IEEEauthorrefmark{2}\IEEEauthorrefmark{3},
        Sijia Yang\IEEEauthorrefmark{4},
        Jingbo Zhou\IEEEauthorrefmark{5},
        Shilei Ji\IEEEauthorrefmark{6},
        Haoyi Xiong\IEEEauthorrefmark{9},~and~
        Dejing Dou\IEEEauthorrefmark{9}
\thanks{\IEEEauthorrefmark{2}Equal contribution. \IEEEauthorrefmark{1}Corresponding author.}
\thanks{\IEEEauthorrefmark{9}J. Liu, H. Xiong, and D. Dou are with Big Data Lab, Baidu Inc., Beijing, China.}
\thanks{\IEEEauthorrefmark{3}L. Mo is with the School of Automation, Southeast University, Nanjing, China.}
\thanks{\IEEEauthorrefmark{4}S. Yang is with the State Key Laboratory of Networking and Switching Technology, Beijing University of Posts and Telecommunications, Beijing, China.}
\thanks{\IEEEauthorrefmark{5}J. Zhou is with Business Intelligence Lab, Baidu Inc., Beijing, China.}
\thanks{\IEEEauthorrefmark{6}S. Ji is with Security department, Baidu Inc., Beijing, China.}
}

\IEEEtitleabstractindextext{%
\begin{abstract}
Due to privacy concerns of users and law enforcement in data security and privacy, it becomes more and more difficult to share data among organizations.  
Data federation brings new opportunities to the data-related cooperation among organizations by providing abstract data interfaces. With the development of \liu{cloud} computing, organizations store data on the \liu{cloud} to achieve elasticity and scalability for data processing. The existing data placement approaches generally only consider one aspect, which is either execution time or monetary cost, and do not consider data partitioning for hard constraints. In this paper, we propose an approach to enable data processing on the \liu{cloud} with the data from different organizations. The approach consists of a data federation platform named \TheName{} and a Lyapunov-based data placement algorithm. \TheName{} enables data processing on the \liu{cloud}. 
We use the data placement algorithm to create a plan in order to partition and store data on the \liu{cloud} so as to achieve multiple objectives while satisfying the constraints based on a multi-objective cost model. The cost model is composed of two objectives, i.e., reducing monetary cost and execution time. We present an experimental evaluation to show our proposed algorithm significantly reduces the total cost (up to 69.8\%) compared with existing approaches.
\end{abstract}
}

\maketitle

\IEEEdisplaynontitleabstractindextext
\IEEEpeerreviewmaketitle


\IEEEraisesectionheading{\section{Introduction}
\label{sec:introduction}}

\IEEEPARstart{D}ata sharing is the first step for the data-related collaborations among different organizations~\cite{Greif1987}, for example, joint modeling with data from multi-party. Meanwhile, direct sharing of raw data with collaborators is difficult due to big volume and/or ownership~\cite{Liu2019,Voigt2017}. Data federation~\cite{federation} virtually aggregates the data from different organizations, which is an appropriate solution to enable data-related collaborations without direct raw data sharing. Based on \liu{cloud} service, data federation works as an intermediate layer to establish an abstract data interface. It provides a virtual data view, on which the involved organizations can collaboratively store, share and process data.

As high efficiency and low cost make it possible to lease resources, e.g.,\ computing, storage, and network, at a large scale, a growing number of organizations tend to outsource their data \liu{onto} the \liu{cloud}. With the pay-as-you-go model, \liu{cloud} computing (\liu{cloud}) brings convenience to the organizations to store and process a large amount of data. Cloud services bring a large number of resources at different layers. 
A Virtual Machine (VM) is an emulator of a computer, which can be viewed as a computing node in a network~\cite{Liu2014}. 
Through the data storage services, unlimited data can be stored on the \liu{cloud}. Cloud providers promise to provide three features, i.e.,\ infinite computing resources available on-demand, dynamic hardware resource provisioning in need, machines and storage paid and released as needed~\cite{above}.
Dynamic provisioning enables \liu{cloud} tenants/users to construct scalable systems with reasonable cost on the \liu{cloud}~\cite{cloud}. With these features, the scientific collaboration on the \liu{cloud} among different organizations becomes a practical solution.

Despite the advantages of \liu{cloud} computing, data security issue on the \liu{cloud} tends to be serious. When the data is stored on the \liu{cloud}, it is crucial to keep confidentiality. Only the authorized tenants/users should have access to the data~\cite{conf}. Encryption is a conventional way to keep the data confidential, such as 
identity-based encryption~\cite{boyen2006anonymous}. 
In addition, the isolation techniques~\cite{guabtni2006customizable}, which provide secure execution spaces for different jobs with specific access controls, are also used to control the accessibility to the data on \liu{the cloud}. A job is composed of a data processing program or a set of data processing programs to be executed on the \liu{cloud} in order to generate new knowledge from the input data. During the scientific collaboration based on the data stored on the \liu{cloud}, the combination of encryption algorithms and isolation techniques can be utilized to keep the confidentiality and security of the data on \liu{cloud}. 

When using the \liu{cloud} services, tenants/users have to pay for them. For instance, when tenants/users directly store their data on the \liu{cloud}, they would be charged for the \liu{cloud} storage service. Widely used \liu{cloud} service providers, such as Amazon Web Services (AWS) \liu{cloud}\footnote{https://aws.amazon.com/}, Microsoft Azure \liu{cloud}\footnote{https://azure.microsoft.com/en-us/} and Baidu \liu{cloud}\footnote{https://cloud.baidu.com/}, provide different data storage types, e.g.,\ hot data storage, data storage with low frequency, cold data storage, and archive data storage, as data storage services. The cost of data storage on the \liu{cloud} varies from type to type. In order to reduce the monetary cost to store and to process the data on the \liu{cloud}, it is necessary to choose a proper data storage type based on a data placement algorithm. 
However, the job execution frequency is not well exploited while constructing the data placement algorithms for the data storage on the \liu{cloud}.
In addition, existing approaches cannot exploit data partitioning techniques to satisfy multiple constraints.

\liu{There are multiple constraints for the data processing on the \liu{cloud}. For instance, when a user requires that the execution of a job should be within a time period, there is a hard time deadline. When a user has a budget limit for the data processing, there is a hard monetary budget for the execution of jobs. In addition, when the system is stable, the jobs can be continuously executed. Otherwise, there may be storage errors during the execution when the accumulated stored data exceed the storage capacity. Thus, the system stability is critical as well.}

In this paper, we propose a solution to enable data processing on the \liu{cloud} for scientific collaboration among different organizations. The solution consists of a secure data processing platform named \TheName{}, a multi-objective cost model, and a Lyapunov-based data placement algorithm. The main contributions of this paper are:
\begin{itemize}
\item The \TheName{} platform. We propose a \liu{cloud} platform, i.e.,\ \TheName{}.
\TheName{} enables secure data processing with the encrypted data stored on the \liu{cloud} for collaboration among different organizations.
\item A data placement problem formulation. We formulate the data placement problem based on a multi-objective cost model and constraints. The multi-objective cost model consists of monetary cost and execution time. The constraints include hard execution time deadline, hard monetary budget, and system stability constraint.
\item A Lyapunov-based data placement algorithm. We use the algorithm to create a data storage plan based on the cost model in order to reduce both monetary cost and the execution time of jobs with the consideration of constraints while exploiting data partitioning techniques. 
\item An extensive experimental evaluation based on a simulation and a widely used benchmark, i.e.,\ Wordcount, and a real-life data processing application for COVID-19~\cite{xiong2020understanding}. The simulation and the experiments are carried out based on a widely used \liu{cloud}, i.e.,\ Baidu \liu{cloud}.
\end{itemize}

The rest of the paper is organized as follows. Section~\ref{sec:relatedWork} reviews related works. Section~\ref{sec:systemDesign} presents the system design of the secure data processing platform. Section~\ref{sec:algorithm} presents the data placement system model, proposes a cost model, shows the hard constraints, and defines the problem. Section~\ref{sec:NODP} proposes the Lyapunove-based data placement algorithm based on the cost model. Section~\ref{sec:exp} shows the experimental results. Finally, we conclude the paper in Section~\ref{sec:con}.

\section{Related Work}\label{sec:relatedWork}

Lyapunov optimization is widely used to optimize the system while ensuring system stability. For instance, Lyapunov optimization is exploited to gain profit~\cite{fang2016stochastic}, to ensure the Quality of Service~\cite{zhang2019integrated} and the time average sensing utility~\cite{wang2018dynamic}. However, the aforementioned work focuses on a single objective besides the system stability and does not consider the task or data partitioning for satisfying multiple constraints. In this paper, we combine the Lyapunov optimization with multiple objectives for data placement.

Data placement is critical to both the monetary cost and the execution time of jobs. In order to reduce the execution time, data transfer can be reduced based on graph partitioning algorithm~\cite{golab2014distributed}. In addition, the data dependency among different jobs can be exploited to reduce the time and monetary cost to transfer data~\cite{Zhao17}. However, these methods only consider one objective, i.e.,\ reducing execution time. They cannot be applied to place the data in different storage types on the \liu{cloud}. A weighted function of multiple costs can be used to achieve multiple objectives, which can generate a Pareto optimal solution~\cite{Liu2016data}, while the authors do not consider the cost to store data on the \liu{cloud} or hard constraints. Load balancing algorithms~\cite{kumar2017dynamic} or dynamic provisioning algorithms~\cite{dynamic} are proposed to generate an optimal provisioning plan in order to minimize the monetary cost while they do not consider the data storage types on the \liu{cloud}. The storage type of the best performance can be selected to store data~\cite{Darwich2020} while the economic storage type can be selected~\cite{Black2016}. However, these two methods cannot address multiple objectives. In this paper, we propose an algorithm to achieve multiple objectives by placing data into various data storage types while satisfying hard constraints.

\liu{In order to handle a multi-objective problem, there are basically two types of solutions, i.e.,\ \textit{a priori} and \textit{a posteriori}~\cite{blagodurov2015multi,Liu2016data}. In this paper, we use an \textit{a priori} method, where the preference information is provided by the users, and then the best solution is produced. Our approach is based on a multi-objective scheduling algorithm focusing on minimizing a weighted sum of objectives. The advantage of such approach is that the scheduling is automatically guided by predetermined weights. In contrast, \textit{a posteriori} methods produce a Pareto front of solutions without predetermined preference information \cite{blagodurov2015multi}. Each produced solution is better than the others with respect to at least one objective, and users need to choose one from the produced solutions, which corresponds to user interference. In this paper, we assume that users can determine the value for the weight of each objective. \textit{a priori} methods can enable us to produce optimal or near-optimal solutions without user interference at run-time. Finally, when the weight of each objective is positive, the minimum of the weighted cost function is already a Pareto optimal solution~\cite{marler2004survey,zadeh1963optimality} and our proposed approach can generate a Pareto optimal or near-optimal solution with the predefined weights. Thus, we do not consider \textit{a posteriori}
methods in this paper.}

Data security is of much importance to the \liu{cloud} users. In order to protect data security, data accessibility is controlled by attributing different levels of permission to avoid unauthorized or malicious access to data on the \liu{cloud}~\cite{security}. In addition, encryption techniques~\cite{Anthes2010, feng2019securegbm} and distributed data storage plan based on data partitioning~\cite{Singh2011, bian2020mp2sda, bian2017multi} can be exploited. Federated learning is proposed to train a model while ensuring data privacy~\cite{google}, yet it is not applicable to the general data processing among different organizations on the \liu{cloud}. In addition, secure separated data processing spaces~\cite{guabtni2006customizable} are proposed to ensure the access control and privacy of data. The separated data processing spaces are disconnected from the public network, which ensures that the confidentiality and the security of data within the local network. Out proposed platform, i.e.,\ \TheName{}, not only provides different data access controls for different tenants/users but also exploits the secure separated data processing spaces to ensure the security and the confidentiality of the data.

\section{System Design}\label{sec:systemDesign}

In this section, we propose a secure data processing platform named \TheName{}. First, we explain the architecture of the platform. Then, we present the life cycle of users' accounts and jobs to be executed.

\subsection{Architecture}
\label{subsec:architecture}

The \TheName{} platform is a data federation platform that provides tenants/users with secure data processing service on the \liu{cloud}. Tenants/users can upload their data onto the platform and execute the self-written programs on Baidu \liu{cloud}. In addition, tenants/users can leverage the data from other organizations for their own data processing jobs, as long as they get permission from the data owners. We illustrate the architecture of the platform and explain the functionalities of each module in this section. As shown in Fig.~\ref{fig:architecture}, the functional architecture of the platform consists of four modules: 
\begin{figure}[t]
\centering
\includegraphics[width=0.45\textwidth]{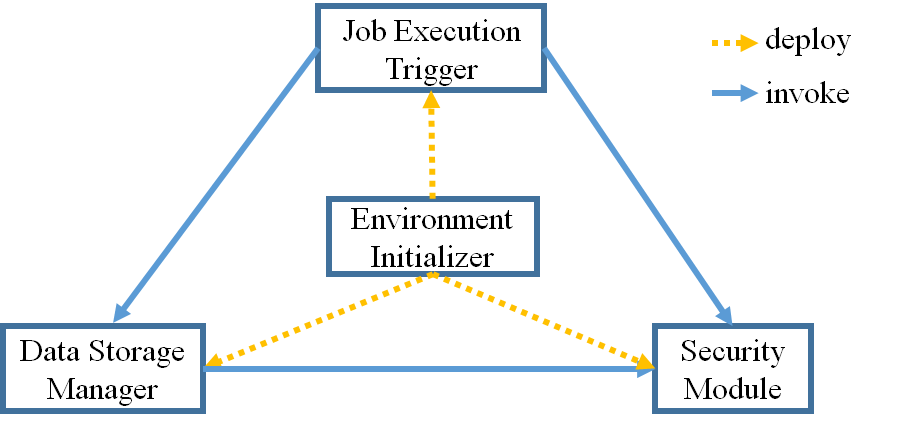}
\vspace{-2mm}
\caption{The functionality architecture of the \TheName{} platform.}
\label{fig:architecture} 
\vspace{-3mm}
\end{figure}

\subsubsection{Environment Initializer}
\label{subsubsec:Initializer}

The environment initializer creates the user account and its execution space on the coordinator node. The created user account is used for the user's security configuration, e.g.,\ the access permission to certain data from another user. The user account is also associated with secure execution spaces for the execution of submitted jobs in the cluster. The secure execution space is a working space without a connection to any public network, which can ensure the confidentiality and the security of the data within the local network.

\begin{figure}[!ht]
\centering
\includegraphics[width=0.45\textwidth]{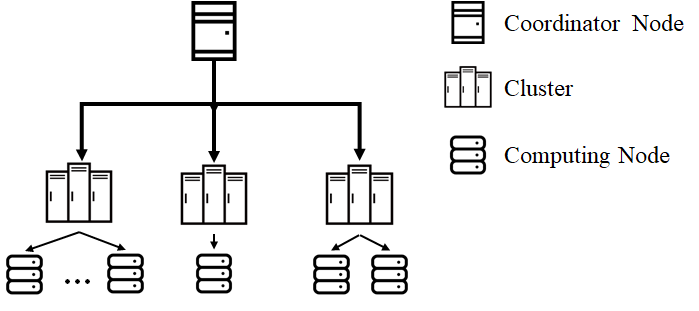}
\vspace{-2mm}
\caption{Infrastructure architecture of the \TheName{} platform.}
\label{fig:node}
\vspace{-3mm}
\end{figure}
As shown in Fig.~\ref{fig:node}, multiple clusters can be dynamically created by the environment initializer module when the execution of jobs is triggered. Each cluster consists of several computing nodes, i.e.,\ VMs on the \liu{cloud}. The coordinator node coordinates the execution among different clusters for all users.
The user has access to the platform through the coordinator node, which is connected to the public Internet. The computing nodes in each cluster are only interconnected with the coordinator node through the local network on the \liu{cloud}. Each computing node is created based on the image~\cite{Wei2009} indicated by the user, which contains necessary tools for the execution of her jobs. An image is a serialized copy of the entire state of a VM stored on the \liu{cloud}~\cite{Wei2009}.

\subsubsection{Data Storage Manager}\label{subsubsec:data_storage_manager}

The data storage manager creates a data storage account and storage buckets on the \liu{cloud} for a user. A storage bucket is a separated storage space to store the data with its own permission strategy. 
The data storage account is used to transfer data between the platform and the user's devices, e.g.,\ computer. Each account is associated with five buckets, i.e.,\ user data bucket, user program bucket, output data bucket, download data bucket, and execution space bucket. Each account has an independent Authorization Key (AK) and Secret Key (SK), with which the tenants/users can send or retrieve the data stored in the buckets.
In addition, the access permission strategy varies from bucket to bucket. For instance, the user has read and write permission on the user data bucket and the user program bucket while she only has the read permission on the download data bucket.
A user can store data in the user data bucket while she can submit self-written codes to the user program bucket. The tenants/users do not have read or write permission on the output data bucket and the execution space bucket. After the execution of the program generated based on the submitted codes, the output data is stored in the output data bucket.
After the confidentiality review of the output data, the output data is transferred to the download data bucket. The review is carried out by the owner of the input data of the job in order to avoid the risk that the raw data or sensitive information appears in the output data of the job. The execution space bucket is used to cache intermediate data of a job, which can be useful for the following execution in order to reduce useless repetitive execution~\cite{Heidsieck2019}.

\subsubsection{Job Execution Trigger}\label{subsubsec:Trigger}

\begin{figure*}[!ht]
\vspace{-5pt}
\centering
\includegraphics[width=0.68\textwidth]{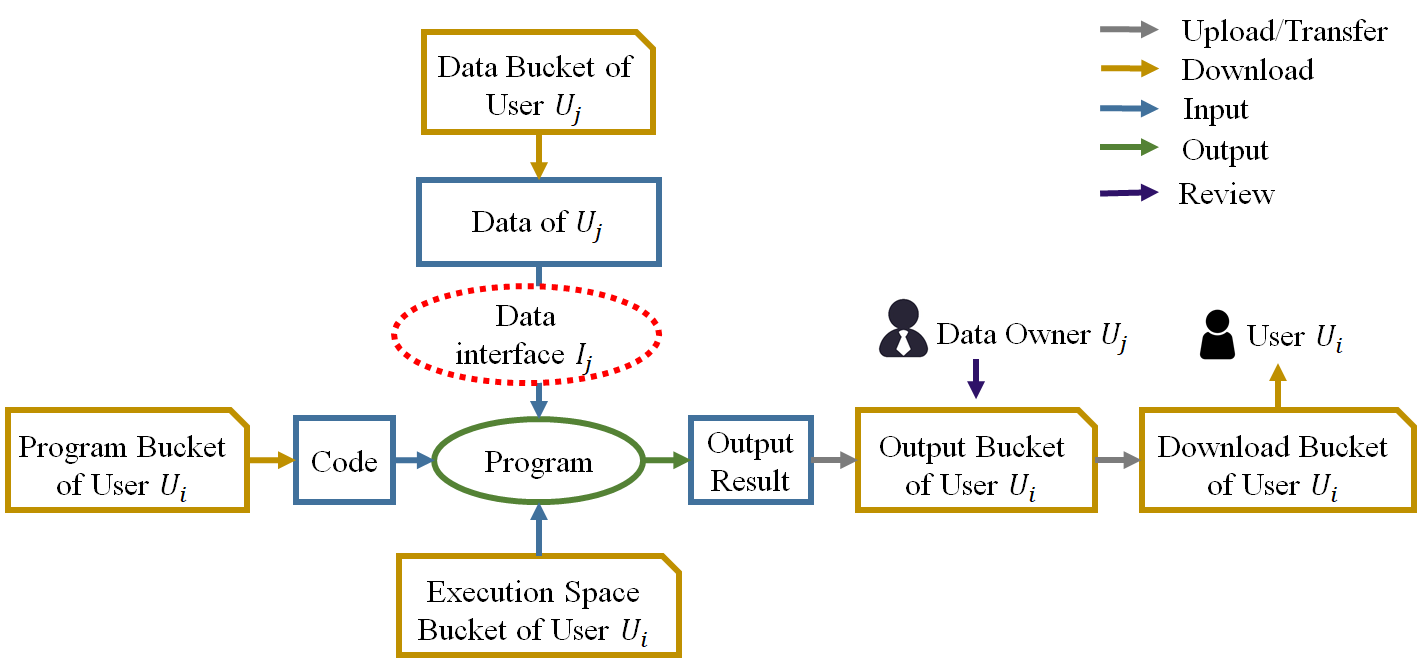}
\vspace{-3mm}
\caption{Job Execution Workflow.}
\vspace{-3mm}
\label{fig:process} 
\end{figure*}
The job execution trigger starts the execution of the job in a cluster. A user can upload the user-written codes onto the platform through a web portal. Then, she can start the execution of the program using the job execution trigger. Once the execution of the program is triggered, a cluster is created, deployed, and configured (see details in Section~\ref{subsubsec:Initializer}). Afterward, the execution of the job is performed in the computing nodes of the cluster.
When several jobs start simultaneously in the same cluster, the job execution trigger creates the same number of execution spaces as that of jobs in order to enable parallel execution without conflict. When the input data of a program consists of the data from other data owners, the corresponding data interfaces are used in order to avoid direct raw data sharing. Let us take two tenants/users as example: User $U_1$ and User $U_2$. A data interface ($I_1$) is defined by the data owner (User $U_1$), which is associated with the data ($D_1$) on the platform. When User $U_2$ gets the permission to use $D_1$, the program generated based on the submitted codes of User $U_2$ can process the data $D_1$ using the Interface $I_1$.
The intermediate data stored in the execution bucket can also be used when the job needs the results of the previous execution.

\subsubsection{Security Module}\label{subsubsec:Security_Module}

In the platform, we use four mechanisms to ensure the security of the data. The first mechanism is to encrypt the data before storing it on the \liu{cloud}. The encryption is based on the Rijndael encryption algorithm~\cite{Jamil2004}.
The second mechanism is to separate the computing nodes from the public network, e.g.,\ Internet, which ensures that no data communication is allowed between the clusters and outside devices, e.g.,\ servers, on the \liu{cloud}. The third mechanism is a uniformed data access control. When a user applies for the permission of the data owned by another user, a data access interface is provided by the data owner instead of direct raw data sharing. The last mechanism is the audition of the codes and output data by data owners, which ensures that no data is leaked from output data.
Through these mechanisms, data confidentiality and security are ensured by the data interface defined by the data owner while ensuring efficient cooperation among different organizations.

\subsection{Life Cycle}\label{subsec:scenario}

In order to present the interactions among users, the platform, and the job execution on the platform, we present the account life cycle and job life cycle. The life cycle describes the state transition of a user account or a job on the platform.  
We assume that there are $n$ scientific collaborators. Each collaborator has private data, which requires keeping confidentiality and security. Through the life cycle, we present how $n$ collaborators process the data on the platform.

\subsubsection{Account Life Cycle}\label{subsubsec:user_account_life_cycle}

The account life cycle consists of three phases, i.e.,\ account creation, data processing, and account cleanup. First, the account related to the user of the platform is created. Then, the user can process the data on the platform. Finally, when the user no longer needs the platform, the data related to the account is removed.

\textit{Account Creation Phase.} 
When a new user needs to use the platform, we create an account and configure the platform using the environment initializer module as shown in Fig.~\ref{fig:architecture}. For the $n$ collaborators in the above scenario, we create $n$ accounts ($U_{t}$ with $t$ representing the number of the collaborator) for each scientific collaborator on the platform. First, the job execution trigger is deployed for each user in the coordinator node. Then, the data storage manager creates a storage account and five storage buckets (see details in Section~\ref{subsubsec:data_storage_manager}) for each user. Afterward, the environment initializer deploys the security module for each user. The security module contains the encryption and decryption information for each user. Please note that the encryption and decryption information is different for different users.

\textit{Data Processing Phase.} 
After the account creation, \liu{data processing jobs can be carried out} on the platform. Before processing the data, each user uploads her own data and the data interface file to the user data bucket. As shown in Fig.~\ref{fig:process}, if User $U_i$ needs to exploit the data from another Users $U_j$, User $U_i$ can apply for the permission. Once User $U_i$ gets the permission from User $U_j$, \liu{the user} also gets the necessary information, e.g.,\ the mock data, to access the data using the corresponding data interface. The mock data contains the data schema of the raw data and some randomly generated examples, while the raw data is never shared with the users. User $U_i$ may use the data from several other tenants/users at the same time. Then, User $U_i$ can submit the codes to process data. In order to process data, User $U_i$ triggers the execution of a job related to the submitted codes (see details in Section~\ref{subsubsec:joblifecycle}), which corresponds to the execution of the job ($j_i$ with $i$ representing the number of the execution) on the platform. During the execution of a job, the intermediate data generated from different execution of the job can be directly used. After the execution and the review of the output data, user $U_i$ can download the output data of Job $j_i$ from the user download bucket. 

\textit{Account Cleanup Phase.} 
When the user no longer needs the platform, the corresponding data, storage buckets, and accounts are removed from the platform by the environment initializer module.

\textit{Initialization Phase.} 
\begin{figure}[!ht]
\centering
\includegraphics[width=0.49\textwidth]{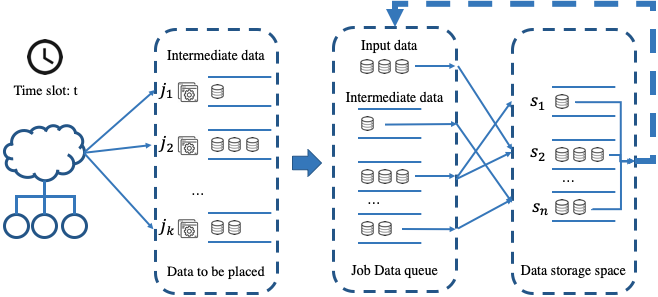}
\caption{System model for data placement.}
\label{fig:dataPlacement}
\end{figure}
\subsubsection{Job Life Cycle}\label{subsubsec:joblifecycle}

The job life cycle consists of four phases, i.e.,\ initialization, data synchronization, job execution, and finalization.

The initialization phase~\cite{mapreduce} is to prepare the environment to execute a job on the platform. The preparation contains three steps: provisioning, deployment, and configuration. First, VMs are provisioned to the job as computing nodes.
There are two cases where existing VMs can be provisioned to the job. The first case is that there are enough live computing nodes on the platform corresponding to the execution of the same or the other jobs of the same user. The second case is that there are enough live computing nodes for the programs of other tenants/users, and all the related tenants/users allow sharing computing nodes. Otherwise, the environment initializer module dynamically creates new VMs as computing nodes, which contain necessary tools for the execution of the job. Then, in order to execute the job, a proper execution space is deployed on the allocated VMs. In order to enable data access, the execution space is configured in each node. For instance, the AK and SK files are transferred into the computing nodes in order to enable data synchronization. 

\textit{Data Synchronization Phase.} 
During the data synchronization phase~\cite{job}, the data storage module synchronizes the data or data interfaces stored on the \liu{cloud}. In addition, the scripts or the files corresponding to the submitted codes are also transferred to the execution space created in the initialization phase. 

\textit{Job Execution Phase.} 
The execution phase~\cite{mapreduce} is the period to execute jobs in the execution space of corresponding VMs. The execution frequency of each job can be dynamically monitored by the platform to compute the cost of data storage. The data, including newly generated intermediate data, is dynamically placed with appropriate storage types with small cost according to the method presented in Section~\ref{sec:NODP}. As shown in Fig.~\ref{fig:process}, after synchronizing the data from buckets, the program corresponding to the submitted codes processes the input data. The execution can be performed in a single computing node or multiple computing nodes in order to reduce the overall execution time. After the execution, the output data is transferred to the output bucket of the user. Once the data is reviewed and approved by the data owners of the input data, it is encrypted by the security module and is transferred to the download bucket to be accessed by the user.

\textit{Finalization Phase.} 
In the finalization phase~\cite{finalization}, the data storage manager uploads the encrypted intermediate of the job. Afterward, the environment initializer module removes the corresponding execution space(s). If a node does not contain any execution space, the node is released, i.e.,\ removed, by the environment initializer, in order to reduce the monetary cost to rent the corresponding VMs.

\section{Multi-Objective Cost Model and Problem Formulation}\label{sec:algorithm}

In this section, we first present the system model for data placement. Then, we propose a cost model based on two costs, i.e.,\ monetary cost and execution time. Afterward, we present the data placement constraints, i.e.,\ hard execution time constraints and the hard monetary budget constraints. Finally, we define the problem to address in the paper.

\subsection{Data Placement System Model}
\label{subsec:systemModel}

The system model for data placement is shown in Fig.~\ref{fig:dataPlacement}. In the \TheName{} platform, we assume that the execution of jobs generates intermediate data at time slot $t$, which may be used as input data in the following time slots, e.g.,\ $t+x$ with $x>0$. Then, the intermediate data should be placed with other input data. Each job has a queue to store the generated intermediate data, and we consider $N$ data storage spaces, which correspond to $N$ storage types with diverse data access speeds and diverse prices to store data. Each data set can be placed to one or multiple data storage types. In order to place a data set to multiple storage types, a data set can be partitioned into several chunks, and each chunk is placed to a data storage type. We assume that the valid time of data set $d_i$ placed at storage type $s_j$ is $T_{max(i, j)}$. If data set $d_i$ is not accessed by any job within $T_{max(i, j)}$, the data set will be removed from the storage space of the platform. When there is a data set generated during the execution of a job or when a job is executed, all the input data is placed again. When the input data is being replaced, its original corresponding storage type is kept until a newly placed storage type is associated.

\subsection{Cost Model}\label{subsec:description}

Inspired by~\cite{Liu2016data}, we propose a multi-objective cost model. 
The cost model is composed of monetary cost and time cost (i.e.,\ the execution time of a job). In order to find a storage plan, we need a cost model to estimate the cost of storing the input data for the execution of jobs. The cost model is generally implemented in the data storage module and under a specific execution environment. In the case of this paper, the execution environment is the \TheName{} platform. The origin of parameters mentioned in this section is summarized in Table~\ref{tal:parameter}. We assume that there are $K$ jobs on the platform.
\begin{table*}[ht]
\centering
\caption{Description of parameters. ``Abbreviation'' represents the abbreviation of the parameters. ``Origin'' represents where the value of the parameter comes from. UD: that the parameter value is  defined by users; Measure: that the parameter value is estimated by the user with the $job$ in a \liu{cloud} environment; Execution: measured during the execution of job in \liu{cloud}; \liu{cloud}: the parameter value is obtained from the \liu{cloud} provider}
\label{tal:parameter}
\begin{tabular}{llll}
\hline
Abbreviation & Parameter & Meaning & Origin \\
\hline
DT & DesiredTime & The estimated execution time to execute a job & UD\\
DM & DesiredMoney & The estimated monetary cost to execute a job& UD\\ 
TDL & TimeDeadline & The hard execution time deadline to execute a job & UD\\
MB & MoneyBudget & The hard monetary budget to execute a job& UD\\  
AIT & averageInitializationTime&The average time to initialize a computing node&Measure\\
CSP & ComputingSpeedPerCPU&The average computing performance of each computing node&Measure\\
WL & workload& The workload of a job&Measure\\
\liuS{$\alpha_k$} & \liuS{$\alpha_k$} & \liuS{The percentage of the workload of Job $k$} that can be executed in parallel&Measure\\
speed & speed&The data transfer speed for a type of data storage & \liu{cloud} \\
SP & StoragePrice&The monetary cost to store data with a storage type & \liu{cloud} \\
RP & ReadPrice&The monetary cost to read data from a cloud storage service & \liu{cloud}\\
VMP & VMPrice & The monetary cost to use a VM & \liu{cloud}\\
\hline     
\end{tabular}
\end{table*}

The total cost to execute a set of jobs with a data placement plan at time slot $t$ is defined as the sum of the total cost of all the jobs:
\begin{equation}
\begin{split}
\mathrm{TotalCost}(Plan[t])=&\sum_{k = 1}^K \mathrm{Cost}(job_k, Plan[t]),
\end{split}
\label{eq:Cost}
\end{equation}
where $Plan[t]$ represents a data placement plan of the data sets related to the set of jobs at time slot $t$, and $j_k$ represents the $k^{th}$ job. 
In the rest of this paper, the total cost represents the normalized cost to execute a set of jobs with a data placement plan per time unit.
$Plan[t]$ is a matrix of data placement variables, which can be expressed by the following formula:
\begin{equation}
\begin{split}
    \mathrm{Plan}[t]=\begin{bmatrix} 
    p_{0,0}[t] & p_{0,1}[t] & ... & p_{0, n}[t] \\ 
    p_{1,0}[t] & p_{1,1}[t] & ... & p_{1, n}[t] \\ 
    ... & ... & \ddots & ... \\ 
    p_{m,0}[t] & p_{m,1}[t] & ... & p_{m, n}[t] \\ 
    \end{bmatrix},
\end{split}
\label{eq:plan}
\end{equation}
where $p_{i,j}[t]$ represents that data set $d_i$ is placed to storage type $s_j$, $m$ represents the number of input and intermediate data sets, and $n$ represents the number of storage types. When $p_{i,j}[t] = 0$, data set $d_i$ is not placed to storage type $s_j$; when $p_{i,j}[t] = 1$, data set $d_i$ is directly placed to storage type $s_j$; When $\liu{0} \leq p_{i,j}[t] \leq 1$, data set $d_i$ is partitioned and the part corresponding to $p_{i,j}[t]$ is placed to storage type $s_j$.

The total cost to execute a job is defined by:
\begin{align}\label{eq:totalCost}
\mathrm{Cost}(job_k, Plan[t])=&~w_{m}\cdot\mathrm{M_{n}}(job_k,Plan[t]))\cdot\mathrm{f}(job_k)\nonumber\\
&~+ w_{t}\cdot\mathrm{T_{n}}(job_k, Plan[t]),
\end{align}
where $\mathrm{T_{n}}(job_k, Plan[t])$ and $\mathrm{M_{n}}(job_k,Plan[t])$ are the normalized time cost and monetary cost, respectively, and they can be defined by Formulas~\eqref{eq:time} and \eqref{eq:money}; $job_k$ represents the $k^{th}$ job and $Plan[t]$ represents the data placement plan at time slot $t$; $w_{t}$ and $w_{m}$ represents the importance of the execution time and the monetary cost of the job. \liu{Defined by the user, $w_{t}$ and $w_{m}$ should be positive values that meet the constraints: $0 \leq w_{t} \leq 1$, $0 \leq w_{m} \leq 1$, $w_{t} + w_{m} = 1$.}
$\mathrm{f}(job_k)$ represents the average frequency of the job execution, which can be dynamically measured according to the history execution before the job execution, e.g.,\ daily, monthly, quarterly and yearly. 
Since the time cost and monetary cost are normalized, neither of them has a unit. 
\liu{Please note that the time cost refers to the execution time of $Jobs$ once while the hard monetary budget is related to the budget per time period, e.g.,\ a day or a month.}

\subsubsection{Time Cost}\label{subsubsec:cost_model}

The normalized time cost is defined by the following formula:
\begin{equation}
\mathrm{T_{n}}(job_k, Plan[t]) = \frac{\mathrm{T}(job_k, Plan[t])}{\mathrm{DT_k}},
\label{eq:time}
\end{equation}
where $\mathrm{T}(j,Plan[t])$ represents the total execution time of the job and $DT_k$ represents the expected execution time (set by the user) of Job $job_k$. \liu{Please note that the execution time $\mathrm{T}(j,Plan[t])$ represents the time to execute $job_k$ once.} The desired execution time could be larger or smaller than the real execution time $\mathrm{Time}(j,Plan[t])$ \liu{while it should be larger than a limit defined by the \TheName ~platform, e.g.,\ $1/20 *\mathrm{SET_{k}}$ with $\mathrm{SET_{k}}$ representing the sequential execution time of Job $job_k$ with one computing node, in order to avoid the unfairness among users.} The total execution time consists of three parts, which are defined by:
\begin{align}\label{exe_time}
\mathrm{T}(job_k,Plan[t]) =&~ \mathrm{InitT}(job_k) \nonumber + \mathrm{DTT}(job_k,Plan[t]) \nonumber\\
&~+ \mathrm{ET}(job_k),
\end{align}

\noindent where $\mathrm{InitT}(job_k)$ represents the Time to Initialize the computing nodes for Job $job_k$; 
$\mathrm{DTT}(job_k,Plan[t])$ represents the Time to Transfer the Data from the \liu{cloud} storage service to computing nodes; 
$\mathrm{ET}(job_k)$ represents the Execution Time of Job $job_k$. 
The initialization of the computing nodes for Job $job_k$ can be specified by the user or realized by the platform, which is out of the scope of this paper while the time can be calculated based on Job $job_k$, e.g.,\ $n_k\cdot AIT$ with $n_k$ representing the number of computing nodes and $AIT$ representing the average time to initialize a computing node. The data transfer time can be calculated based on the size of the input data of Job $job_k$ and the data placement plan as follows:
\begin{equation}\label{eq:storagePlan}
\mathrm{DTT}(job_k,Plan[t]) =
\sum_{j = 1}^N\sum_{i\in \mathrm{data_k}}\dfrac{size(d_i)}{speed_j}\cdot p_{i, j}[t],
\end{equation}
where $speed_j$ represents the speed to transfer data from data storage type $j$ to computing nodes. 
\liu{As explained in Section \ref{subsec:systemModel}, $N$ represents the total number of storage types on the \TheName ~platform.}
According to the Amdahl's law~\cite{Sun2010}, the execution time of Job $j$ can be estimated by the following formula~\cite{Liu2016data} \liu{when the impact of data communication can be ignored with $n$ computing nodes}:
\begin{equation}
\mathrm{ET}(job_k)=\frac{[\liuS{\alpha_k}/n+(1\liu{-}\liuS{\alpha_k})]\cdot\mathrm{WL}(job_k)}{CSP},
\label{eq:executionTime}
\end{equation}
where \liuS{$\alpha_k$}\footnote{\liu{
\liuS{$\alpha_k$} can be obtained by measuring the execution time of executing the job $k$ twice with different numbers of \liuS{computing nodes}~\cite{Liu2016data}. For instance, assume that we have $t_1$ for \liuS{$m_1$} \liuS{computing nodes} and $t_2$ for \liuS{$m_2$} \liuS{computing nodes}, 
\begin{equation}\label{eq:peq1101}
\begin{split}
\liuS{\alpha_k} = \frac{\liuS{m_2} * \liuS{m_1} * ( t_2 - t_1 )}{ \liuS{m_2} * \liuS{m_1} * ( t_2 - t_1 ) + \liuS{m_1} * t_1 - \liuS{m_2} * t_2 }
\end{split}
\end{equation}}} 
represents the percentage of the workload that can be executed in parallel; $n$ is the number of computing nodes, which is configured by users \liu{while it should be less than a limit, e.g.,\ 20, defined by the \TheName ~platform}; $\mathrm{WL}(j)$ represents the workload of a job which can be measured by the number of FLOP (FLoat-point Operations)~\cite{2014evaluating}. $CSP$ is the average computing performance of each computing node, which can be measured by the number of FLOPS (FLoating-point Operations Per Second). 

\subsubsection{Monetary Cost}\label{subsubsec:monetary_cost}

Normalized monetary cost is defined by the following formula:
\begin{equation}
\mathrm{M_{n}}(job_k, Plan[t]) = \frac{\mathrm{M}(job_k, Plan[t])}{\mathrm{DM_k}},
\label{eq:money}
\end{equation}
where $\mathrm{Money}(job_k,Plan[t])$ is the financial cost to rent VMs as computing nodes on the \liu{cloud}. \liu{Please note that the monetary cost $\mathrm{Money}(job_k,Plan[t])$ represents the total monetary cost to execute $job_k$ within a time period, e.g.,\ a month.} $DM_k$ represents the expected execution monetary cost of Job $Job_k$, which can be configured by the user \liu{while it should be larger than or equal to a limit defined by the \TheName ~platform, i.e.,\ $\mathrm{SMC_{k}}$ with $\mathrm{SMC_{k}}$ representing the monetary cost of Job $job_k$ with one computing node, in order to avoid the unfairness among users.}
$DM_k$ can be bigger or smaller than the real monetary cost $\mathrm{Money}(job_k, t)$. 
$\mathrm{Money}(job_k, Plan[t])$ can be estimated based on the following formula:
\begin{align}\label{monetaryCost}
\mathrm{M}(job_k,Plan[t]) =&~ \mathrm{EM}(job_k, Plan[t]) \nonumber\\ 
&+ \mathrm{DSM_k}(job_k, Plan[t]) \nonumber\\ 
&+ \mathrm{DAM_k}(job_k, Plan[t]),
\end{align}
where $\mathrm{EM}(job_k, Plan[t])$ represents the Monetary cost to use the computing nodes to Execute the job; $\mathrm{DSM}(job_k, Plan[t])$ represents the Monetary cost to Store the Data on the \liu{cloud} storage service; $\mathrm{DAM}(job_k, Plan[t])$ represents the Monetary cost to Access to the Data. $\mathrm{EM}(job_k, Plan[t])$ can be estimated by the following formula:
\begin{multline}
\mathrm{EM}(job_k, Plan[t]) = \mathrm{VMP}(job_k) \cdot n_k\\
\cdot[\mathrm{T}(job_k, Plan[t]) - \mathrm{InitT}(job_k)],
\end{multline}
where $\mathrm{VMP}(job_k)$ represents the average monetary cost of a VM for the execution of Job $j$; $n_k$ represents the number of computing nodes to execute the job; $\mathrm{T}(job_k, Plan[t])$ and $\mathrm{InitT}(job_k)$ are defined in Formula~\eqref{exe_time}.

We allocate the storage monetary cost of a data set to the jobs based on the workload. $\mathrm{DSM}(j,Plan[t])$ is defined by the following formula:
\begin{multline}
\mathrm{DSM}(job_k,Plan[t])=
\frac{\mathrm{WL}(job_k)}{\sum_{l=1}^K(\mathrm{WL}(job_l) \cdot \mathrm{f}(job_l))}\cdot \\
\sum_{j=1}^N\sum_{i\in\mathrm{data_k}}(\mathrm{SP_j} \cdot \mathrm{size}(d_i) \cdot p_{i, j}[t]),
\end{multline}
where $\mathrm{WL}(job_k)$ represents the workload of job $job_k$; $\mathrm{dataset}(j)$ represents the data sets that job $j$ uses; $\mathrm{job}(i)$ represents the jobs that takes data $i$ as input data; $\mathrm{SP}(s_i)$ represents the monetary cost to store the data with the storage type $s_i$, which is defined in the data placement plan $plan[t]$, on the \liu{cloud}; $\mathrm{size}(d_i)$ represents the size of the input data $d_i$. 

$\mathrm{DAM}(job_k, Plan[t])$ is defined by the following formula:
\begin{multline}\label{eq:dataaccess}
\mathrm{DataAccessMoney}(job_k, Plan[t]) = \\  
\sum_{j=1}\sum_{i\in\mathrm{data_k}}(\mathrm{RP_j}\cdot\mathrm{size}(d_i)\cdot p_{i, j}[t]),
\end{multline}
where $\mathrm{RP}_j$ represents the monetary cost to read data $d_i$ from the \liu{cloud} storage service; $\mathrm{size}(d_i)$ represents the size of the input data $d_i$ of Job $job_k$.

\subsection{Data Placement Constraints}\label{sec:constraint}

In this section, we present the constraints of data placement. First, we present the hard execution time and monetary budget constraints for each job. Then, we present the system stability constraint based on Lyapunov optimization.

We assume that there are hard time deadline and hard monetary budget for each job, which can be formulated as follows:
\begin{equation}
\mathrm{T}(job_k, Plan[t])\leq\mathrm{TDL_k},~\forall k \in [0, K],
\label{eq:hardTime}
\end{equation}
\begin{equation}
\mathrm{M}(job_k, Plan[t])\leq\mathrm{MB_k},~\forall k \in [0, K],
\label{eq:hardBudget}
\end{equation}
where \liu{$\mathrm{T}(job_k, Plan[t])$ and $\mathrm{M}(job_k, Plan[t])$ are defined in Formulas~\ref{exe_time} and \ref{monetaryCost} respectively,} $TDL_k$ represents the hard execution time deadline, $MB_k$ represents the hard monetary cost Budget, and $Jobs$ represents the set of jobs in the system. 

For storage spaces, as shown in the right part of Fig.~\ref{fig:dataPlacement}, we use $S_j(t)$ to denote the set of data sets placed in the data storage space of Type $j$. Therefore, the dynamic set is defined as follows:
\begin{equation}
\mathrm{S}_j(t+1)=\max[\mathrm{S}_j(t) - r_j(t), 0] + \sum_{i=1}^{M}p_{i,j}[t],
\label{eq:storageQueue}
\end{equation}
where $r_j(t)$ represents the data to be removed because of time limit and $\sum_{i=1}^{m}p_{i,j}[t]$ represents the newly placed data sets to data storage type $s_j$.

For jobs shown in the middle part of Fig.~\ref{fig:dataPlacement}, we use $J_i$ to denote the set of data sets generated from the execution of Job $i$. We have the following job data storage set defined as follows:
\begin{equation}\label{eq:jobQueue}
\mathrm{J}_k(t+1)=\max\left[\mathrm{J}_k(t) - \sum_{j = 1}^N\sum_{i \in data_k} p_{i,j}[t], 0\right] + G_k[t], 
\end{equation}
where $data_k$ represents the set of input data sets of Job $k$, and $G_k[t]$ represents the newly generated intermediate data of Job $k$.

We exploit the Lyapunov optimization technique~\cite{Lyapunov2017} by considering both the set of data sets placed in the data storage spaces and the job data storage sets. Let $D(t) = (S_j(t), J_i(t), j \in \{1,\ldots,n\}, i \in \{1,\ldots,k\}, t \in \{1, 2, \ldots,\})$ denote all the data sets in time slot $t$. We have the following constraint in order to ensure the stability of the system:
\begin{equation}
\Bar{D} \triangleq \lim_{T \rightarrow \infty} \frac{1}{T} \sum_{t=0}^{T-1}\left(\sum_{j=1}^{N}\E\{S_j(t)\} + \sum_{k=1}^{K}\E\{J_k(t)\}\right) < \infty.
\label{eq:stable}
\end{equation}

\subsection{Problem Definition}
\label{problem_definition}
\begin{table*}[ht]
\centering
\caption{The monetary cost to store data on the \liu{cloud} with different storage types, i.e.,\ Standard, Low frequency, Cold and Achieve.}
\label{tal:storage}
\begin{tabular}{lllll}
\hline
& Standard & Low frequency & Cold & Achieve \\
Expected data access frequency & frequently & $<$ once per month & $<$ once per year & $\geq$ three years \\
Cost to store data (Dollar/GB/month) & 0.0155& 0.0113& 0.0045 & 0.015\\  
Cost to read data (Dollar/GB) & N/A & 0.0042 & 0.0085 & 0.12 \\
\hline     
\end{tabular}
\end{table*}

The problem we address in the paper is a data placement problem, i.e.,\ how to choose a storage type to store the data in order to reduce \liu{the expected total cost}, which consists of the monetary cost and the execution time of jobs, while satisfying constraints, on the \liu{cloud}. A job can be executed multiple times because the user-defined codes are updated, or the parameters are updated~\cite{Chunduri2018}.
As shown in Table~\ref{tal:storage}, different data storage types of storage services on the \liu{cloud} correspond to different prices. The storage type with higher expected data access frequency, e.g.,\ Standard, has a higher price and higher data access speed.
The total cost to execute a job once differs with different data placement plans.
Thus, the problem we address in this paper is how to find an optimal data placement plan of all the data sets in order to reduce \liu{the expected total cost} to execute the jobs with different execution frequencies based on a cost model. We define \liu{the expected total cost} as:
\begin{equation}
\overline{Cost}(Jobs, Plan[t]) = \lim_{T \rightarrow \infty} \dfrac{1}{T}\sum_{t=0}^{T-1}\E\{Cost(Jobs, Plan[t])\}.
\label{eq:cost}
\end{equation}
Then, the problem addressed in this paper can be formulated as follows:
%
\begin{align}\label{eq:problem}
&\displaystyle\min~\overline{Cost}(Jobs, Plan[t])\\
&~\text{s.t.}
\begin{cases}
\displaystyle p_{i,j}[t] \in [0, 1] \nonumber\\
\displaystyle\sum_{j=1}^{N} p_{i,j}[t] = 1,\nonumber\\
\displaystyle\mathrm{Formulas}~\eqref{eq:hardTime}, 
~\liu{\eqref{eq:hardBudget}}, ~\mathrm{and}~\eqref{eq:stable}.\nonumber\\
\end{cases}
\end{align}
where $Jobs$ represents the set of jobs in the system.

\liuS{The data placement problem is a typical NP-hard problem~\cite{Ren2018}. Let us separate the data placement variables into two parts, i.e., $p'_{i,j}$ and $p''_{i,j}$. $p'_{i,j}$ is a continuous variable between 0 and 1, which represents the partitioning of a data set. $p''_{i,j}$ is a 0-1 integer, which is the scheduling decision. Then, we have $p_{i,j} = p'_{i,j} * p''_{i,j}$. Then, the problem defined in Formula \ref{eq:cost} is a Mixed Integer Linear Programming (MILP) problem, which is a proven NP-hard problem \cite{mo2018controllable, burer2012non, hartmanis1982computers}. In this case, the exhaustive search for an optimal solution for $p''_{i,j}$ increases exponentially and the complexity is $\mathcal{O}(N^M)$, which cannot be solved within a polynomial time, with $N$ representing the number of storage types and $M$ represents the number of data sets. }

\section{Near-Optimal Data Placement}\label{sec:NODP}

In this section, we present a near-optimal data placement approach based on Lyapunov optimization. Lyapunov optimization is widely used to achieve optimization objectives while ensuring the system stability~\cite{Lyapunov2017,Liu2020Lyapunov}. In order to exploit Lyapunov optimization techniques, we first construct a Lyapunov function and propose a Lyapunov-based algorithm (LNODP) to perform the data placement while ensuring system stability. Then, we propose a greedy approach to perform the near-optimal data placement while satisfying hard deadlines.
The greedy approach consists of three algorithms, i.e.,\ near-optimal data planning (NOD Planning), near-optimal data placement, and data placement (NOD Placement) with partitioning (NOD Partitioning). LNODP exploits NOD Planning to generate a near-optimal data placement plan; NOD Planning takes advantage of NOD Placement to choose optimal data storage type when the hard constraints can be satisfied, and NOD Placement uses NOD partitioning to generate a plan to partition the data in order to satisfy hard constraints when one data storage type does not work.

\subsection{Lyapunov Optimization based Data Placement}

We define a Lyapunov function $L(t)$ as follows:
\begin{equation}
L(t)\overset{\Delta}{=}\dfrac{1}{2}\left(\sum_{j=1}^{N}[S_j(t)]^2+\sum_{k=1}^{K}[J_k(t)]^2\right).
\label{eq:lyapunov}
\end{equation}
This function represents the data sets to be placed. Then, we can define the derivative of the Lyapunov function as follows:
\begin{equation}
\dfrac{\triangle L(t)}{\triangle t}\overset{\Delta}{=}\E\{L(t+\triangle t)- L(t)|D(t)\}.
\label{eq:derivative}
\end{equation}
We use the expectation to address the randomness of the intermediate data generated by the execution of jobs and the data placement actions. As to solve the problem defined in Formula~\eqref{eq:problem} requires the global information the \TheName ~system, which is hard to predict or gather, \liu{we transform} the problem defined in Formula~\eqref{eq:problem} \liu{to} the following objective function, \liu{which is a greedy conversion with limited local information}:
\begin{align}\label{eq:objective}
&\displaystyle\min\left(\dfrac{\triangle L(t)}{\triangle t} + \omega \cdot \E\{Cost|D(t)\}\right)\\
&~\text{s.t.}
\begin{cases}
\displaystyle\mathrm{T}(job_k, Plan[t]) < \mathrm{TDL_k}, \forall k \in [1, K]\nonumber\\
\displaystyle\mathrm{M}(job_k, Plan[t]) < \mathrm{MB_k}, \forall k \in [1, K]\nonumber\\
\displaystyle p_{i,j}[t] \in [0, 1],\nonumber
\end{cases}
\end{align}
where \liu{$\mathrm{T}(job_k, Plan[t])$ and $\mathrm{M}(job_k, Plan[t])$ are defined in Formulas~\ref{exe_time} and \ref{monetaryCost} respectively, $TDL_k$ represents the hard execution time deadline, $MB_k$ represents the hard monetary cost Budget, and} the parameter $\omega\geq 0$ represents the importance of the \liu{expected} total cost compared with the stability of the system.
\begin{theorem}\label{th1}
The objective function has the following upper bound when $\triangle t = 1$: \\
\begin{align}
&\triangle L(t) + \omega\E\{Cost(Jobs, Plan[t])|D(t)\} \nonumber\\
&\leq L + \omega\sum_{k=1}^K C_k\nonumber\\
&+\E\left\{\sum_{j=1}^N\mathrm{S}_j(t)r_j(t))|D(t)\right\} - \E\left\{\sum_{k=1}^K\mathrm{J}_k(t)G_k[t])|D(t)\right\}\nonumber\\
&+\E\left\{\sum_{j=1}^N\sum_{k=1}^K\sum_{i \in data_k} (\mathrm{J}_k(t) - \mathrm{S}_j(t) + \omega C'_{i,j})p_{i,j}[t]|D(t)\right\},
\label{eq:bound}
\end{align}
with $L$ defined in Formula~\eqref{eq:abre}, $C$ defined in Formula~\eqref{eq:cValue}, and $C'$ defined in Formula~\eqref{eq:cPrimeValue}. 
\label{theorem:upperBound}
\end{theorem}

\begin{proof}
First, we focus on the data stored in the job data queue with the assumption that $\sum_{i=1}^{M}p_{i,j}[t]\leq d_{\max}$ and $r_j(t)\leq r_{\max}$:
\begin{align}
&~\mathrm{S}^2_j(t+1) - \mathrm{S}^2_j(t) \nonumber\\
&=\left(\max[\mathrm{S}_j(t) - r_j(t), 0] + \sum_{i=1}^{M}p_{i,j}[t]\right)^2 - \mathrm{S}^2_j(t) \nonumber\\
&\leq \left(\sum_{i=1}^{M}p_{i,j}[t]\right)^2 + (r_j(t))^2 - 2\mathrm{S}_j(t)\left(r_j(t) - \sum_{i=1}^{M}p_{i,j}[t]\right) \nonumber\\
&\leq (d_{\max})^2 + (r_{\max})^2 - 2\mathrm{S}_j(t)\left(r_j(t) - \sum_{i=1}^{M}p_{i,j}[t]\right).
\label{eq:queueDelta}
\end{align}

Then, we have the similar results for the data storage spaces with the assumption that $\sum_{i\in data_k, j = 1}^N p_{i,j}[t]\leq data_{\max}$, where $data_{\max}$ represents the maximum number of data sets for any job, and $G_k[t]\leq G_k^{\max}$:
\begin{align}
&~\mathrm{J}^2_k(t+1) - \mathrm{J}^2_k(t) \nonumber\\
=&~\left(\max[\mathrm{J}_k(t) - \sum_{i \in data_k, j = 1}^N p_{i,j}[t], 0] + G_k[t]\right)^2 - \mathrm{J}^2_k(t) \nonumber\\
\leq &~(G_k[t])^2 + \left(\sum_{i\in data_k, j = 1}^N p_{i,j}[t]\right)^2 \nonumber\\
&~- 2\cdot\mathrm{J}_k(t)\left(\sum_{i \in data_k, j = 1}^N p_{i,j}[t] - G_k[t]\right) \nonumber\\
\leq &~(G_k^{\max})^2 + (data_{\max})^2 \nonumber\\
&~- 2 \cdot \mathrm{J}_k(t)\left(\sum_{i \in data_k, j = 1}^N p_{i,j}[t] - G_k[t]\right) 
\label{eq:dataDelta}
\end{align}

With Formulas~\eqref{eq:queueDelta} and \eqref{eq:dataDelta}, we have:
\begin{align}
&\triangle \{L(t)|D(t)\} \nonumber\\
& \leq  L +\E\left\{\sum_{j=1}^N\mathrm{S}_j(t)(r_j(t) - \sum_{i=1}^{M}p_{i,j}[t])|D(t)\right\} \nonumber\\ 
& + \E\left\{\sum_{k=1}^K\mathrm{J}_k(t)\left(\sum_{i \in data_k, j = 1}^N p_{i,j}[t] - G_k[t]\right)|D(t)\right\},
\label{eq:upperProof}
\end{align}
\begin{align}
L=\frac{N}{2}\cdot[(d_{\max})^2 + (r_{max})^2]+\frac{K}{2}\cdot[(G_k^{\max})^2 
+(data_{\max})^2].
\label{eq:abre}
\end{align}

\begin{algorithm}[tbp] 
\caption{Lyapunov-based Near-Optimal Data Placement}
\label{alg:lyapu} 
\begin{algorithmic}[1]
\INPUT
$D$: A set of data sets; \newline
\quad $T$: Maximum number of iterations; \newline
\quad $T'$: Maximum number of iterations for generating data placement plans; \newline
\quad $Plan[t]$: data placement plan in Time slot $t$.
\OUTPUT $Plan[t+1]$: data placement plan in Time slot $t + 1$. 
\State D $\leftarrow$ sort(D) \label{alg:LBA:sort}
\For{t $\in$ T} 
    \While{not all data sets $\in D$ are placed and $iter < T'$}
        \State $Plan^*[t] \leftarrow $ NearOptimalDataPlanning($Plan[t]$) \label{alg:LBA:planning}
        \For{each Data set $d_i$ in $D$} \label{alg:LBA:foreachset}
            \For{$j \in N$}
                \State $p^*_{i, j}[t] \leftarrow$ getPlan($Plan^*[t]$, $i$, $j$)
                \If{$C'_{i, j} \leq 0$ and $p^*_{i, j}[t] \neq 0$} \label{code:lba:ifset}
                    \State Set $p_{i, j}[t+1] = p^*_{i, j}[t]$ \label{code:lba:set}
                \Else
                    \State Set $p_{i, j}[t+1] = 0$ \label{code:lba:unset}
                \EndIf
            \EndFor
        \EndFor
    \EndWhile
    \State $iter \leftarrow iter + 1$
\EndFor
\State Update $J_k(t)$ and $S_j(t)$
\end{algorithmic}
\end{algorithm}

The cost model presented in Section~\ref{subsec:description} can be rewritten as:
\begin{multline}
\E\{\mathrm{cost}(Jobs, Plan[t])|D(t)\} = \\ 
\sum_{k=1}^K C_k +\E\left\{\sum_{j=1}^N\sum_{k=1}^K\sum_{i \in data_k} C'_{i,j,k} \cdot p_{i,j}[t]|D(t)\right\},
\label{eq:costSum}
\end{multline}
\begin{align}
C_k = &~\Bigg(\dfrac{\omega_t \cdot n_k \cdot \mathrm{AIT}}{\mathrm{DT_k}} +\left(\dfrac{\omega_t}{\mathrm{DT_k}}+\dfrac{\omega_m\cdot \mathrm{VMP}(job_k)\cdot n_k}{\mathrm{DM_k}}\right) \nonumber\\
&~\cdot \frac{(\frac{\liuS{\alpha_k}}{n}+(1+\liuS{\alpha_k}))*\mathrm{WL}(job_k)}{\mathrm{CSP}}\Bigg)\cdot\mathrm{f}(job_k),
\label{eq:cValue}
\end{align}
\begin{align}
C'_{i,j,k} = &~\Bigg(\dfrac{\omega_t}{speed_j \cdot \mathrm{DT_k}} + \dfrac{\omega_m \cdot \mathrm{VMP}(job_k) \cdot n_k}{speed_j \cdot \mathrm{DM_k}} + \frac{\omega_m \cdot \mathrm{RP_j}}{\mathrm{DM_k}}\nonumber\\
&~+ \frac{\omega_m \cdot \mathrm{WL}(job_k) \cdot \mathrm{SP_j}}{\sum_{l=1}^K (\mathrm{WL}(job_k) \cdot \mathrm{f}(job_k)) \cdot \mathrm{DM_k}} \Bigg)\nonumber\\
&~\cdot size(d_i) \cdot \mathrm{f}(job_k) 
\label{eq:cPrimeValue}
\end{align}

Finally, we can take the expectation and add the total cost, i.e.,\ $\mathrm{cost}(Plan[t])$ to both sides of Formula~\eqref{eq:upperProof} and hence Theorem~\ref{theorem:upperBound} is proven.
\end{proof}

In order to solve the problem defined in Formula~\eqref{eq:problem}, we minimize the upper bound of Theorem~\ref{eq:bound}. As the status of Time slot $t$ can be observed in the system, we only need to minimize the following element:
\begin{align}
&~\E\left\{\sum_{j=1}^N\sum_{k=1}^K\sum_{i \in data_k} (\mathrm{J}_k(t) - \mathrm{S}_j(t)+\omega C'_{i,j,k})p_{i,j}[t]|D(t)\right\} \nonumber\\
=&~\E\left\{\sum_{j=1}^N\sum_{i=1}^M C'_{i,j}p_{i,j}[t]|D(t)\right\},
\label{eq:finalObjective}
\end{align}
with $C'_{i,j}$ defined as:
\begin{equation}
C'_{i,j} = \sum_{k \in Jobs_i} (\mathrm{J}_k(t) + \omega C'_{i,j,k}) - \mathrm{S}_j(t),
\end{equation}
where $Jobs_i$ represents the set of jobs that process Data set $d_i$.

We design a Lyapunov-based approach to minimize Formula~\eqref{eq:finalObjective}, as shown in Algorithm~\ref{alg:lyapu}. First, we sort the data sets based on $C'_{i,j}$ in descent order in order to minimize the cost of the data set corresponding to high costs first (Line~\ref{alg:LBA:sort}). 
Then, for each data set, we use Algorithm~\ref{alg:Storage} to find an optimal data placement plan (Line~\ref{alg:LBA:planning}).
For each combination of $\{i, j\}$ (Line~\ref{alg:LBA:foreachset}), if $C'_{i,j} \leq 0$ (Line~\ref{code:lba:ifset}), we will update $p_{i,j}[t+1] = p^*_{i,j}$ (Line~\ref{code:lba:set}), otherwise, we will set $p_{i,j}[t+1] = 0$ (Line~\ref{code:lba:unset}).
\liuS{Please note that the data set is placed with a data placement plan that meets reasonable constraints based on Algorithm~\ref{alg:Storage} when $C'_{i,j} \leq 0$. Otherwise, the placement plan remains idle and will be set with a proper data placement plan in later time intervals. When the data placement plan of a data set is idle, the execution of related jobs is postponed until the placement plan is set in order to meet the constraints.}
\liu{Please note that when the data in the system exceed the capacity of the system or the \liuS{given constraints are not reasonable}, the algorithm may not generate a proper data placement plan that meets all the constraints.}

\subsection{Near-Optimal Data Placement Algorithm}\label{data_placement}
\begin{algorithm}[tbp] 
\caption{Near-Optimal Data Planning}
\label{alg:Storage} 
\begin{algorithmic}[1]
\INPUT $D$: A set of data sets; \newline
$Plan$: data placement plan in Time slot $t$.
\OUTPUT $Plan^*$: The near-optimal data placement plan of each data $d$ in data set $D$ with the minimum cost.
\For{each Data $d$ in $D$}
\State $\mathrm{cost_{before}}\leftarrow$ calculateCost($Plan$) \Comment{According to Formula~\eqref{eq:Cost}} \label{al:nodp:cost}
\State $Plan'\leftarrow$ getNearOptimalPlacement($d$, $Plan$) \label{al:nodp:getplan}
\If {$\mathrm{Cost(Plan')}<\mathrm{Cost(Plan)}$} \label{code:fram:com} \label{al:nodp:updatebegin}
    \State $Plan^*\leftarrow Plan'$
\EndIf
\EndFor \label{al:nodp:updateend}
\end{algorithmic}
\end{algorithm}

Based on the multi-objective cost model, we propose a greedy algorithm to generate a near-optimal data placement plan while reducing the \liu{expected} total cost to execute a set of jobs on the \TheName{} platform as shown in Algorithm~\ref{alg:Storage}. In the algorithm, for each Data set $d$, we first calculate the total cost based on the cost model (Line~\ref{al:nodp:cost}). Then, we generate a near-optimal data placement plan by replacing Data set $d$ while keeping the other data sets based on Algorithm~\ref{alg:placement} (Line~\ref{al:nodp:getplan}). Afterward, if the new data placement plan can reduce the total cost according to the cost model, we update the data placement plan if the new data placement plan corresponds to a smaller total cost (Lines~\ref{al:nodp:updatebegin} - \ref{al:nodp:available}).
\begin{algorithm}[tbp] 
\caption{Near-Optimal Data Placement}
\label{alg:placement} 
\begin{algorithmic}[1]
\INPUT $d$: A data set; \newline
$Jobs$: A set of jobs that process Data set $d$; \newline
$StorageTypeList$: The list of storage types; \newline
$Plan$: a data placement plan.
\OUTPUT $Plan^*$: The near-optimal data placement plan of Data set $d$.
\State $Plan^* \leftarrow Plan$
\State $j^*$ = getOptimalType($Plan$, $d$, $StorageTypeList$) \label{alg:NODP:chooseType}
\State $TypesForTimeConstraints \leftarrow$ getTypesForTimeConstraint($Plan$, $d$, $Jobs$) \label{al:nodp:time}
\State $TypesForMonetaryConstraints \leftarrow$ getTypesForMonetaryConstraint($Plan$, $d$, $jobs$) \label{al:nodp:money}
\State $AvailableTypes \leftarrow TypesForTimeConstraints \cap TypesForMonetaryConstraints$ \label{al:nodp:available}
\If {$j^* \in AvailableTypes$}
    \State For $j \in [1, N]$ and $p_{i, j} \in Plan^*$, set $p_{i,j} = \left\{
\begin{aligned}
1 & , j = j^* \\
0 & , j \neq j^*
\end{aligned}
\right.$
\Else
    \State $Plan^* \leftarrow$ dataPlacementWithPartitioning($d$, $Plan$, $Jobs$, $TypesForTimeConstraints$, \\$TypesForMonetaryConstraints$)
\EndIf
\end{algorithmic}
\end{algorithm}

Algorithm~\ref{alg:placement} replaces Data set $d$ in order to reduce the total cost. First, we choose an optimal data storage type $j^*$ based on the data placement plan by trying each data storage type in $storageTypeList$ (Line~\ref{alg:NODP:chooseType}). Then, we choose a set of possible storage type candidates that meet both the hard time deadline constraint and hard monetary budget constraint (Lines~\ref{al:nodp:time} - \ref{al:nodp:money}). If the chosen data storage type $j^*$ is within the set of storage type candidates, we will update the data storage placement. If not, we will exploit Algorithm~\ref{alg:partitioning} to place the data set with data partitioning.
\begin{algorithm}[tbp] 
\caption{Data Placement With Partitioning}
\label{alg:partitioning} 
\begin{algorithmic}[1]
\INPUT $d$: A set of data; \newline
$Jobs$: A set of job that process Data set $d$; \newline
$StorageTypeList$: The list of storage types; \newline
$TypesForTimeConstraints$: A set of storage types that only meet the hard execution time constraint; \newline
$TypesForMonetaryConstraints$: A set of storage types that only meet the hard monetary budget constraint; \newline
$Plan$: a data placement plan.
\OUTPUT $Plan^*$: The near-optimal data placement plan of each data $d$; \newline
$Feasibility$: If there is a data placement plan that meets the two constraints
\State $Plan^* \leftarrow Plan[t]$
\If {$TypesForTimeConstraints = \emptyset$ or $TypesForMonetaryConstraints = \emptyset$} \label{alg:dpwp:verify}
    \State $Feasibility = False$ \label{alg:dpwp:unavailable}
\Else
    \State $j_1 \leftarrow$ getOptimalTypeForTimeConstraint($Plan$, $d$, $Jobs$, $TypesForTimeConstraints$)  \label{alg:dpwp:time}
    \State $j_2 \leftarrow$ getOptimalTypeForMonetaryConstraint($Plan$, $d$, $Jobs$, $TypesForMonetaryConstraints$) \label{alg:dpwp:money}
    \State possibleArea $\leftarrow$ [0, 1] \label{alg:dpwp:areabegin}
    \For{j $\in$ [1, N]} 
        \State possibleArea $\leftarrow$ possibleArea $\cap$ getArea($Plan$, $d$, $j_1$, $j_2$, $Jobs$)
    \EndFor \label{alg:dpwp:areaend}
    \If{possibleArea $= \emptyset$} \label{alg:dpwp:areaVerify}
        \State $Feasibility = False$ \label{alg:dpwp:areaEmpty}
    \Else
        \State $p \leftarrow$ getOptimalPart($d$, $plan$, $Jobs$, $possibleArea$) \label{alg:dpwp:getpart}
        \State For $j \in [1, N]$ and $p_{i, j} \in Plan^*$, 
        set \\ $p_{i, j} = \left\{
\begin{aligned}
p & , j = j_1, \\
1-p & , j = j_2, \\
0 & , else
\end{aligned}
\right.$ \label{alg:dpwp:updatepart}
    \EndIf
\EndIf
\end{algorithmic}
\end{algorithm}

Algorithm~\ref{alg:partitioning} generates a near-optimal data placement plan with the consideration of data partitioning while meeting the two constraints, i.e.,\ the hard time deadline constraint and the hard monetary budget constraint. First, if any of the set of available data storage type candidates for the hard time deadline constraint or hard monetary budget constraint is an empty set, we consider that the two constraints cannot be met (Lines~\ref{alg:dpwp:verify} and \ref{alg:dpwp:unavailable}). If not, first, we choose an optimal type ($j_1$ for the time constraint and $j_2$ for the monetary constraint)  within the set of candidates for each constraint by trying each storage type (Lines~\ref{alg:dpwp:time} and \ref{alg:dpwp:money}). We define a possible area as the range of parts of the data set to be placed at Type $j_1$ while meeting both the two constraints. We can calculate the possible area for each job and the intersection of the area for all the related jobs (Lines~\ref{alg:dpwp:areabegin} - \ref{alg:dpwp:areaend}). 
Given a related job $job_k$ of a data set and two data storage types ($j_1$, $j_2$), we can calculate the possible area based on Formulas~\eqref{eq:Cost} - \eqref{eq:dataaccess}, \eqref{eq:hardTime} and \eqref{eq:hardBudget}, and the calculated area is: $\max\{0, a\} \leq p_{i, j_1} \leq \min\{b, 1\}$ when $c > 0$, or $\max\{a, b\} \leq p_{i, j_1} \leq 1$ when $c <0$, with:
\begin{align}
a =&~ \frac{TDL_k - ET(job_k) - n_k \cdot AIT}{size(d)} \nonumber\\
&~\cdot \frac{speed_{j_1} \cdot speed_{j_2}}{speed_{j_2} - speed_{j_1}} - \frac{speed_{j_1}}{speed_{j_2} - speed_{j_1}},\nonumber
\end{align}
\begin{align}
b =&~ \frac{MB_k}{c \cdot size(d)} - \frac{VMP(job_k) \cdot n_k \cdot ET(job_k)}{c \cdot size(d)}\nonumber\\
&~- \frac{VMP(job_k) \cdot n_k}{c \cdot speed_{j_2}} - \frac{SP_{j_2}}{c \cdot size(d)} - \frac{RP_{j_2}}{c \cdot size(d)},\nonumber
\end{align}
\begin{align}
c =&~ VMP(job_k)\cdot n_k\cdot\left(\frac{1}{speed_{j_1}}-\frac{1}{speed_{j_2}}\right)\nonumber\\
&~+ d \cdot (SP_{j_1} - SP_{j_2}) + (RP_{j_1} - RP_{j_2}),\nonumber
\end{align}
\begin{equation}
d=\frac{\mathrm{WL}(job)}{\sum_{l=1}^K (\mathrm{WL}(j_l) \cdot \mathrm{f}(j_l))},\nonumber
\end{equation}
where $AIT$ represents the average initialization time, $ET$ represents the execution time, which can be calculated based on Formula~\ref{eq:executionTime}, $SP$ represents the storage price, $RP$ represents the read price.
Finally, if the final possible area is an empty set, we consider that the two constraints cannot be met (Lines~\ref{alg:dpwp:areaVerify} and \ref{alg:dpwp:areaEmpty}). 
If not, we calculate the optimal data partitioning by choosing a boundary of the area that corresponds to a smaller total cost and update the data placement plan (Lines~\ref{alg:dpwp:getpart} and \ref{alg:dpwp:updatepart}).

\subsection{Algorithm Analysis}

Let us assume that we have $M$ input data, $N$ data storage types, and each input data is related to $K$ jobs on average. Then, the search space for the problem we address is $\mathcal{O}(N^M)$, which is the complexity of the brute-force method. The complexity of ActGreedy algorithm~\cite{Liu2016data} is $\mathcal{O}(M * K * N)$. Then, the complexity of LNODP is $\mathcal{O}(T * M * K * N)$ (when there is no need to execute Algorithm~\ref{alg:partitioning}) or $\mathcal{O}(T * M * K^2 * N)$ (when Algorithm~\ref{alg:partitioning} is executed for each job), which is much smaller than $\mathcal{O}(N^M)$ when $N^{M-1} > M * K^2 * T$ (this is a general case). \liu{Please note that we do not reduce the complexity of the problem but reduce the complexity of the solution.} The complexity of Economic and Performance (see details in Section~\ref{sec:exp}) is $\mathcal{O}(M * M)$. Although the complexity of LNODP is slightly bigger than that of ActGreedy, Economic, or Performance, it can generate near-optimal data placement plans while satisfying hard constraints. 

LNODP can generate a near-optimal result while satisfying the hard constraints in most cases. However, there are two cases where LNODP cannot generate a data placement plan to satisfy hard constraints for a job. First, when there is no data storage type to store all the input data of a job while satisfying both the hard time deadline and the hard monetary budget. Second, when there is no combination of two storage types that can satisfy both the hard time deadline and the hard monetary budget. In these two cases, the user should reset the hard constraints of the job in order to use LNODP to generate data placement plans.

\liu{In order to analyze the worst case guarantee of the LNODP algorithm, we focus on the case when the data is scheduled according to the near-optimal data plan as explained in Line~\ref{code:lba:set} of Algorithm~\ref{alg:lyapu}. The case explained in Line~\ref{code:lba:unset} of Algorithm~\ref{alg:lyapu} is ignored as the data is not scheduled in this case. The problem addressed in Algorithm \ref{alg:Storage} is a scheduling problem when there is an optimal solution according to Algorithms~\ref{alg:placement} and \ref{alg:partitioning}. When the data can be scheduled without being partitioned, the solution is equal to the solution generated by a greedy algorithm. As the cost function of the combination problem is monotone, Algorithm~\ref{alg:partitioning} can generate an optimal combination of the chosen storage types by Algorithm~\ref{alg:placement}. Thus, when the data needs to be partitioned while scheduling, the solution is also equal to the solution generated by a greedy algorithm. As the scheduling problem while minimizing a cost function is a typical submodular problem as explained in \cite{fokkink2019submodular}, the worst case guarantee of the LNODP algorithm becomes the worst case guarantee of a greedy algorithm for a submodular, which is $\frac{e-1}{e} * f^*$, where $e$ is the base of the natural logarithm and $f^*$ represents the optimal solution~\cite{nemhauser1978analysis}.}

\begin{figure}[!ht]
\centering
\includegraphics[width=0.4\textwidth]{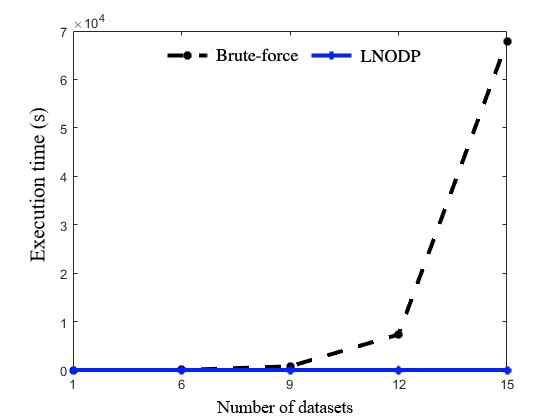}
\vspace*{-3mm}
\label{sub:figexe}
\caption{Execution time of Greedy and Brute-force}
\label{fig:Comparison1}
\vspace*{-6mm}
\end{figure}

\section{Experimentation}\label{sec:exp}

In this section, we first present the simulation to compare the execution time of our proposed Lyapunov-based Near-Optimal Data Placement (LNODP) algorithm and the brute-force method. We consider four storage types, i.e.,\ Standard, Low frequency, Cold, and Archive, in our proposed algorithm. These four storage types are provided by the storage service on the Baidu \liu{cloud}. Then, we compare the total cost of four storage methods: LNODP, brute-force, Performance~\cite{Darwich2020}, and Economic~\cite{Black2016}. 
The brute-force method is to search the minimum cost in the entire searching space, which means that the result of brute-force is the optimal solution. The Performance method~\cite{Darwich2020} uses the storage type that corresponds to the highest data transfer speed. Economic~\cite{Black2016} uses the storage type that corresponds to the smallest price to store data. In addition, we compare our algorithm with a simple adapted greedy algorithm, i.e.,\ ActGreedy~\cite{Liu2016data}, to show that our algorithm can address multiple hard constraints while ActGreedy only reduces the total cost without considering the hard constraints. 
Then, we present the comparison of the total cost among the four storage methods using a widely used data processing benchmark, i.e.,\ Wordcount on Hadoop~\cite{White2015}, and a real-life data processing program for the correlation analysis of COVID-19~\cite{xiong2020understanding} (COVID-19-Correlation), which is selected from recent work related to COVID-19~\cite{xiong2020understanding,liu2020analysis, liu2020investigation}.
In the experimentation, we consider five execution frequencies (daily, semimonthly, monthly, quarterly, and yearly) for Wordcount and COVID-19-Correlation.

\subsection{Simulation}\label{subsec:Simulation}

In this section, we compare our proposed algorithm with the brute-force method in terms of the execution time and the total cost. We take 15 data sets with the average size being 5.5 GB as the input data of jobs. We execute fifteen jobs to process the input data. Each job is associated with different data sets, including Wordcount, Grep, etc. Each job is with different frequencies and different settings such as $DT$, $w_t$. The data sets include DBLP XML files \cite{dblp} and some data sets from Baidu. The DBLP XML file contains the metadata, e.g.,\ the name of authors, publishers, of computer-based English articles. The comparison experiment results are shown in Fig.~\ref{fig:Comparison1}.

\begin{figure}[!ht]
\centering
\includegraphics[width=0.4\textwidth]{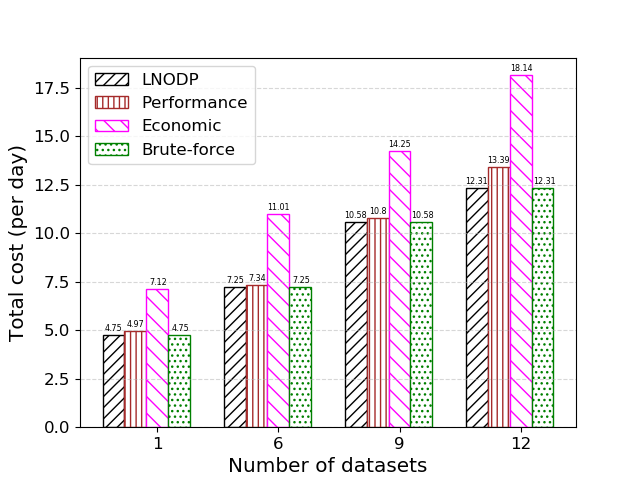}
\label{sub:fig-total}
\vspace*{-3mm}
\caption{Comparison among four methods}
\label{fig:Comparison2}
\vspace*{-3mm}
\end{figure}

Fig.~\ref{fig:Comparison1} shows the result of the execution time of different methods. In order to generate a data placement plan for six data sets with fifteen jobs, the execution time of the greedy algorithm is shorter than 0.0001s, while that of LNODP is 0.08s. When the number of the data sets augments, the execution time of the brute-force method increases exponentially. When the number of data sets becomes 15, the execution time of the brute-force method is 67839s, while that of LNODP remains within 0.0001s. 

Fig.~\ref{fig:Comparison2} presents the comparison among four methods: LNODP, brute-force, Performance, and Economic. LNODP corresponds to the same total cost as that of the brute-force method, which is up to 8.2\% and 30.6\% smaller than that of Performance and Economic, respectively. The simulation experiment shows that the result of our proposed algorithm is as same as the brute-force method, which means the result of our proposed algorithm is the optimal solution in these situations.
\begin{figure*}[htbp]
\centering
\subfigure[Daily]{
\includegraphics[width=0.31\textwidth]{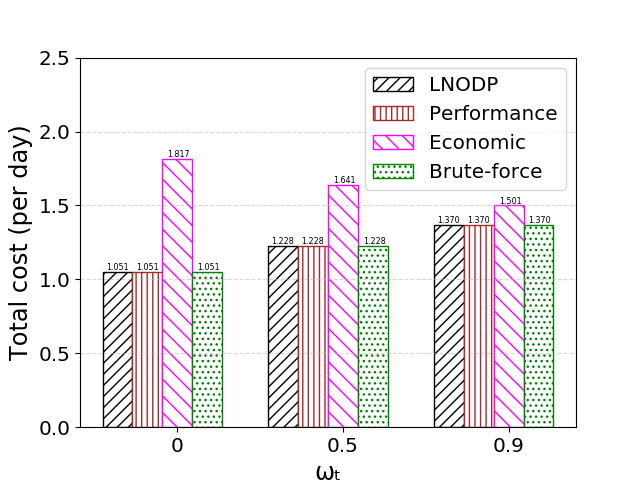}
\label{sub:daily}
}
\subfigure[Quarterly]{
\includegraphics[width=0.31\textwidth]{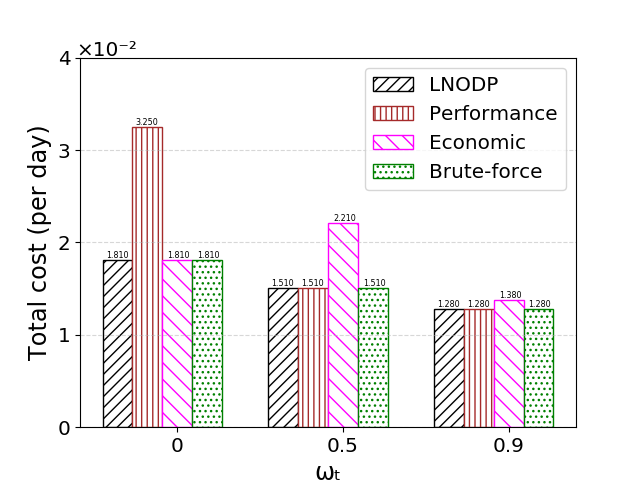}
\label{sub:quarter}
}
\subfigure[Yearly]{
\includegraphics[width=0.31\textwidth]{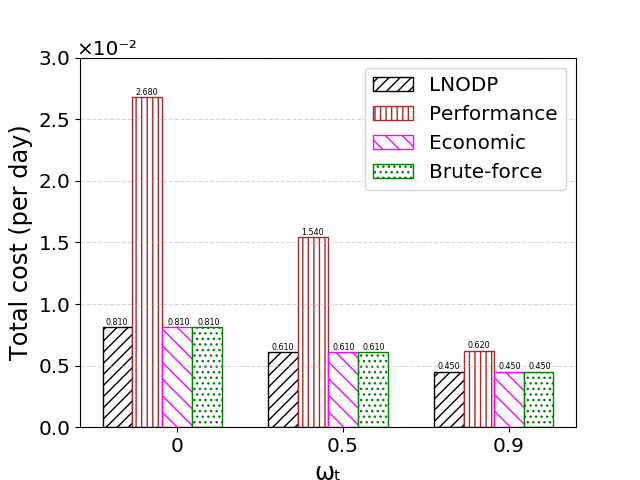}
\label{sub:yearly}
}
\vspace*{-3mm}
\caption{Total cost of Wordcount}
\vspace*{-4mm}
\label{fig:wordcount}
\end{figure*}
\begin{figure*}[htbp]
\centering
\subfigure[Daily]{
\includegraphics[width=0.31\textwidth]{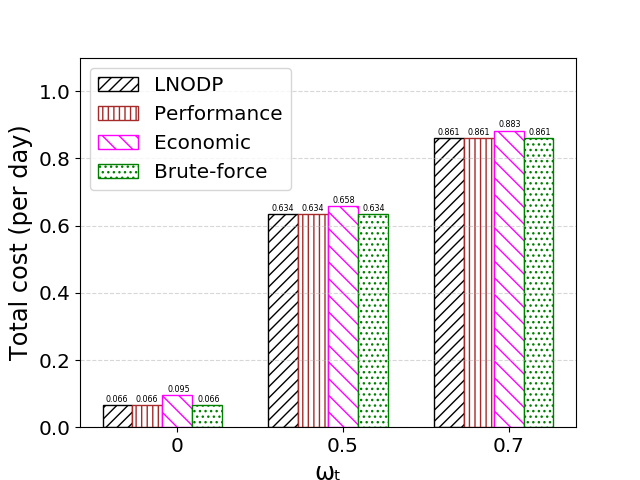}
\label{sub:covid-daily}
}
\subfigure[Quarterly]{
\includegraphics[width=0.31\textwidth]{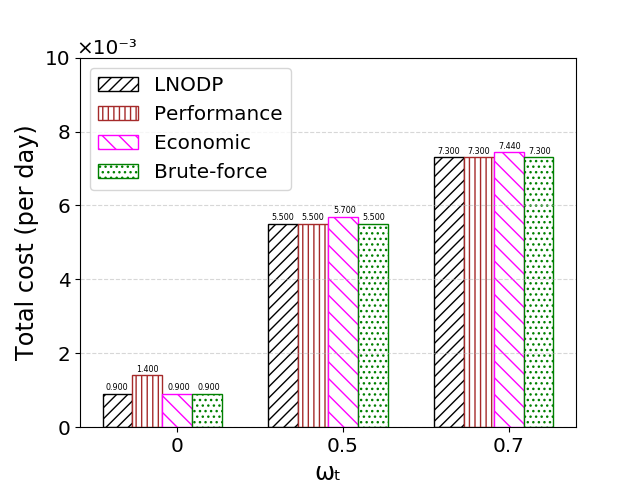}
\label{sub:cocid-quarter}
}
\subfigure[Yearly]{
\includegraphics[width=0.31\textwidth]{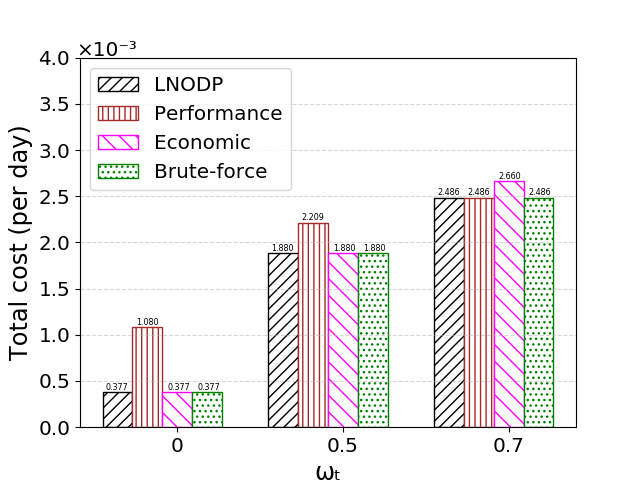}
\label{sub:covid-yearly}
}
\vspace*{-3mm}
\caption{Total cost of COVID-19-Correlation.}
\vspace*{-6mm}
\label{fig:covid}
\end{figure*}

\subsection{Wordcount}
\label{subsec:wordcount}

Hadoop~\cite{hadoop} is a framework for parallel big data processing on a cluster of commodity servers. Hadoop contains two components, i.e.,\ HDFS~\cite{Shvachko2010} and MapReduce. HDFS is a distributed file system with a master-slave architecture. MapReduce is a programming model and implementation for parallel data processing in a distributed environment. MapReduce contains two phases, i.e.,\ Map and Reduce. In the Map phase, the input data is processed, and key-value pairs are generated. In the reduce phase, the key-value pairs of the same Key are processed. 

Wordcount is a widely used benchmark, which counts the frequency of each word in the input files. Wordcount contains two steps, i.e.,\ Map and Reduce. In the Map step, $<word, 1>$ is generated for each work in the input data. Then, the number of $<word, 1>$ is counted for each work in the Reduce step. Finally, the frequency of each word is calculated and stored in HDFS.

We deploy Hadoop on three computing nodes based on the platform. Each node is a VM with one CPU core and 4 GB RAM. We use DBLP 2019 XML files of 6.04 GB as the input data.
We set $DT$ as 1200 seconds and $DM$ as 1 dollar.

First, we set the hard time deadline as 2000 seconds and 10 dollars. Fig.~\ref{fig:wordcount} shows that our proposed algorithm, i.e.,\ LNODP, significantly outperforms the baseline approach. When the frequency is daily, the total cost corresponding to different approaches is shown in Fig.~\ref{sub:daily}. Compared with Economic, LNODP can reduce the total cost by 42.2\%, \liuS{25.2}\%, and \liuS{8.7}\% when $\omega _{t}$ is \liuS{0}, 0.5, and \liuS{0.9} respectively. When the frequency is quarterly, LNODP can reduce the total cost by \liuS{44.3}\% compared with Performance when the $\omega _{t}$ is 0 as Fig.~\ref{sub:quarter} shows. LNODP can generate an optimal storage plan, which significantly outperforms (the total cost is \liuS{31.7}\% and \liuS{7.2\%} smaller) Economic when $\omega _{t}$ is 0.5 \liuS{and 0.9, respectively}. Fig.~\ref{sub:yearly} presents the efficiency of our proposed algorithm when the frequency is yearly. Compared with Performance, our algorithm can reduce the total cost by \liuS{69.8\%,} 60.4\% and 27.4\% when $\omega_{t}$ is \liuS{0,} 0.5 and 0.9, respectively. 

\begin{table}[]
        \centering
        \caption{Results for hard execution time constraint and hard monetary budget constraint. Frequency: yearly. Hard time deadline: 1420; hard monetary budget: 6.5. The time unit is second and the monetary unit is yuan.}
        \begin{tabular}{|c|c|c|c|c|}
                \hline
                \multirow{2}{*}{} & \multicolumn{2}{c|}{Constraints} & \multirow{ 2}{*}{Cost}& \multirow{2}{*}{\textbf{$\omega_t$}} \\\cline{2-3}
                & Time & Monetary & & \\ \hline
                LNODP & Satisfied (1420.0) & Satisfied (6.5) & 0.018 & \multirow{4}{*}{0}  \\ \cline{1-4}
                ActGreedy & \textbf{Broken (1465.8)} & Satisfied (2.9) & 0.0081 &   \\ \cline{1-4}
                Performance & Satisfied (1405.4) & \textbf{Broken (9.7)} & 0.027 &  \\ \cline{1-4}
                Economic & \textbf{Broken (1465.8)} & Satisfied (2.9) & 0.0081 & \\  \hline \hline
                LNODP & Satisfied (1420.0) & Satisfied (6.5) & 0.0053 & \multirow{ 4}{*}{0.9}  \\ \cline{1-4}
                ActGreedy  & \textbf{Broken (1465.8)} & Satisfied (2.9) & 0.0045 &   \\ \cline{1-4}
                Performance & Satisfied (1405.4) & \textbf{Broken (9.7)} & 0.0062 &  \\ \cline{1-4}
                Economic & \textbf{Broken (1465.8)} & Satisfied (2.9) & 0.0045 & \\  \hline
        \end{tabular}
        \label{tab:wordcount}
\end{table}

Fig.~\ref{fig:wordcount} presents that our algorithm can reduce the total cost by up to \liuS{69.8}\% compared with Performance and up to \liuS{42.2}\% compared with Economic. As the execution frequency of the job decreases, the advantage of our algorithm becomes significant. The comparison of Fig.~\ref{sub:daily}, \ref{sub:quarter} and \ref{sub:yearly} indicates that as \liuS{the importance of time cost becomes bigger}, i.e.,\ $\omega _{t}$, increases, the advantage of our proposed algorithm becomes significant as well. This experiment also shows that the result of our algorithm can generate the optimal solution as the brute-force method.

Table~\ref{tab:wordcount} presents the execution with a strict hard execution time constraint and a hard monetary budget constraint, i.e.,\ 1420 seconds and 6.5 dollars. The existing methods, e.g.,\ ActGreedy, Performance, Economic, cannot meet both the two constraints, while LNODP can place the data with data partitioning while satisfying the two hard constraints with small total cost. In addition, we find that the weight of objectives only impacts the total cost, which has no impact on the satisfaction of the constraints.

\begin{table}[]
        \centering
        \caption{Results for hard execution time constraint and hard monetary budget constraint. Frequency: yearly. Hard time deadline: 722; hard monetary budget: 1.9. The time unit is second and the monetary unit is yuan.}
        \begin{tabular}{|c|c|c|c|c|}
                \hline
                \multirow{2}{*}{} & \multicolumn{2}{c|}{Constraints} & \multirow{ 2}{*}{Cost}& \multirow{2}{*}{\textbf{$\omega_t$}} \\\cline{2-3}
                & Time & Monetary & & \\ \hline
                LNODP & Satisfied (722.0) & Satisfied (1.8) & 0.0050 & \multirow{4}{*}{0}  \\ \cline{1-4}
                ActGreedy  & \textbf{Broken (732.1)} & Satisfied (0.7) & 0.0019 &   \\ \cline{1-4}
                Performance & Satisfied (720.8) & \textbf{Broken (1.95)} & 0.00054 &  \\ \cline{1-4}
                Economic & \textbf{Broken (732.1)} & Satisfied (0.7) & 0.0019 & \\  \hline \hline
                LNODP & Satisfied (722.0) & Satisfied (1.8)  & 0.0038 & \multirow{ 4}{*}{0.7}  \\ \cline{1-4}
                ActGreedy  & \textbf{Broken (732.1)}  & Satisfied (0.7)  & 0.0029 &   \\ \cline{1-4}
                Performance & Satisfied (720.8) & \textbf{Broken (2.0)}  & 0.0040 &  \\ \cline{1-4}
                Economic & \textbf{Broken (732.1)} & Satisfied (0.7) & 0.0029 & \\  \hline
        \end{tabular}
        \label{tab:covid}
\end{table}

\liu{In addition, the average execution time of LNODP, Performance, Economic, and Brute-force are 2.79$*10^{-4}$, 4.26$*10^{-5}$, 4.14$*10^{-5}$, 2.98$*10^{-4}$, respectively. While LNODP corresponds can generate good data placement plans, the execution time remains quite acceptable.}

\subsection{COVID-19}\label{subsec:COVID-19}

Since the coronavirus disease (COVID-19) has become a global emergency, we reproduced the data processing program for the correlation among COVID-19-related search activities, human mobility, and the number of confirmed cases in Mainland China presented in~\cite{xiong2020understanding}. The data involved in~\cite{xiong2020understanding} includes the number of confirmed cases in each city ($dataset_c$), the volume of COVID-19-related search activities in each city ($dataset_s$), inflows and outflows for each city ($dataset_m$) and the  population in each city ($dataset_p$). $dataset_m$ is the inflow and outflow data of inter-city population with the transitions of the inter-city mobility categorized by the origin and destination pairs. $dataset_s$ includes the keywords and phrases related to the epidemic from January to March.
The total amount of these data sets is 1.134 GB. 

The data processing for the COVID-19-related correlation analysis consists of the following three steps. First, the data is selected using a filter operation. Then, a join operator is used to generate the features for each city, i.e.,\ the number of confirmed cases, the inflows, the outflows, the search volumes, the population. Afterward, the correlation between any two features is calculated for each city. The experimental results are shown in Fig.~\ref{fig:covid}. We set $DT$ as 600 seconds and $DM$ as 0.5 dollars.

First, we set the hard execution time constraint as 800 seconds and the hard monetary budget constraint as 2 dollars. Fig.~\ref{fig:covid} shows that our proposed algorithm, i.e.,\ LNODP, significantly outperforms Performance (up to \liuS{65.1}\%) when the size of the input data of the job is smaller than that of Wordcount. When the frequency is daily, the total costs of different approaches are shown in Fig.~\ref{sub:covid-daily}. When $\omega_{t}$ is 0 and $\omega_{m}$ is 1, our algorithm can reduce the total cost by 30.5\% compared with Economic. When $\omega _{t}$ increases to 0.5, our algorithm can reduce the total cost by \liuS{3.6}\% compared with Economic. When the importance of time, i.e.,\ $\omega_{t}$, increases to 0.7, our algorithm can outperform Economic, and the total cost can be reduced by \liuS{2.5}\%.  
Fig.~\ref{sub:cocid-quarter} presents the total cost of different approaches when the frequency is quarterly. When the user only considers the importance of money, our algorithm can reduce the total cost by 35.7\% compared with Performance. With the increase of $\omega _{t}$, our algorithm can reduce the total cost by \liuS{3.5}\% and 1.9\% compared with Economic when $\omega_{t}$ is 0.5 and 0.7 respectively.
When the frequency is yearly, the execution results are shown in Fig.~\ref{sub:covid-yearly}. The most significant result is that our algorithm can reduce the total cost by \liuS{65.1}\% compared with Performance when $\omega_{t}$ is 0 and $\omega_{m}$ is 1. When $\omega_{t}$ is 0.5 and 0.7, our algorithm can reduce the total cost by \liuS{14.9}\% and \liuS{6.5}\% compared with Performance and Economic, respectively. 

From Fig.~\ref{fig:covid}, we find that our proposed algorithm, i.e.,\ LNODP, significantly outperforms the Performance method (up to \liuS{65.1}\%) and the Economic method (up to \liuS{30.5}\%), when the frequency of the job execution is high and when the size of the input data of the job is big. 
%

Table~\ref{tab:covid} presents the execution with a strict hard execution time constraint and a hard monetary budget constraint, i.e.,\ 722 seconds and 1.9 dollars. 
The existing methods, e.g.,\ ActGreedy, Performance, Economic, cannot meet both the two constraints. However, LNODP can place the data with data partitioning while satisfying the two hard constraints with a small total cost. We find that the weight of objectives only impacts the total cost while having no impact on the satisfaction of the constraints.

\liu{In addition, the average execution time of LNODP, Performance, Economic, and Brute-force are 2.01$*10^{-4}$, 2.01$*10^{-5}$, 2.01$*10^{-5}$, 2.04$*10^{-4}$, respectively. The execution time of LNODP is quite acceptable.}


\section{conclusion}
\label{sec:con}

When organizations outsource their data \liu{onto} the \liu{cloud}, it is critical to choose a proper data placement strategy to reduce \liu{its expected total} cost. In this paper, we proposed a solution to enable data processing on the \liu{cloud} with the data from different organizations. The approach consists of three parts: a data federation platform with secure data sharing and secure data computing, a multi-objective cost model, and a Lyapunov-based near-optimal data placement algorithm. The cost model consists of monetary cost and execution time. The Lyapunov-based near-optimal algorithm \liu{delivers a solution to the problem}. We carried out extensive experiments to validate our proposed approach. The experimental results indicate that our proposed algorithm outperforms the baseline approaches up to 69.8\% and that our algorithm can generate the same optimal solution as the brute-force method within a short execution time.



\bibliographystyle{IEEEtran}

\bibliography{references}

\end{document}